\documentclass{article}
\usepackage{amsmath,amssymb,color,url,amsthm,cases}
\usepackage[pdftex,hiresbb]{graphicx}
\usepackage{type1cm}
\usepackage{subcaption}
\usepackage{geometry}
\usepackage{here}
\allowdisplaybreaks
\newtheorem{thm}{Theorem}[section]

\newtheorem{lem}[thm]{Lemma}
\newtheorem{prop}[thm]{Proposition}
\newtheorem{ass}[thm]{Assumption}
\newtheorem{ex}[thm]{Example}

\newtheorem{assA}{Assumption}
\theoremstyle{definition}
\newtheorem{defn}[thm]{Definition}
\newtheorem{rem}[thm]{Remark}
\def\l     {\left}
\def\r     {\right}
\def\<     {\langle}
\def\>     {\rangle}
\def\fin   {\hfill{$\Box$}\vspace{5mm}}
\def\d     {\displaystyle}
\def\bbC   {{\mathbb C}}
\def\bbE   {{\mathbb E}}
\def\bbN   {{\mathbb N}}
\def\bbP   {{\mathbb P}}
\def\bbR   {{\mathbb R}}
\def\calB  {{\mathcal B}}

\def\calD  {{\mathcal D}}
\def\calE  {{\mathcal E}}
\def\calF  {{\mathcal F}}

\def\ve    {\varepsilon}

\def\olnu    {\overline{\nu}}

\def\tH    {\widetilde{H}}
\def\tN    {\widetilde{N}}

\def\olM {\overline{M}}
\def\whf {\widehat{f}}
\def\whg {\widehat{g}}
\def\Dom   {\mathop{\mathrm{Dom}}\nolimits}
\def\dequiv {\overset{\mathrm{d}}{=}}
\begin{document}
\title{Local risk-minimization for exponential additive processes}
\author{Takuji Arai\footnote{Department of Economics, Keio University, 2-15-45 Mita, Minato-ku, Tokyo, 108-8345, Japan. \\ (arai.z8@keio.jp)}}
\maketitle

\begin{abstract}
We explore local risk-minimization, a quadratic hedging method for incomplete markets, in exponential additive models.
The objectives are to derive explicit mathematical expressions and to conduct numerical experiments.
While local risk-minimization is well studied for L\'evy processes, little is known for the additive process case
because, unlike L\'evy processes, the L\'evy measure for an additive process depends on time, which significantly complicates the mathematical framework.
This paper shall provide a set of necessary conditions for deriving expressions for LRM strategies in exponential additive models, as integrability conditions on the L\'evy measure,
which allow us to confirm whether these conditions are satisfied for given concrete models.
In the final section, we introduce the variance-gamma scaled self-decomposable process, a Sato process that generalizes the variance-gamma process, as a primary example,
and perform numerical experiments. \\
{\bf Keywords:} Local risk-minimization, Additive processes, Malliavin-Skorohod calculus, Carr-Madan method.
\end{abstract}

%
%
\section{Introduction}\setcounter{equation}{0}
Local risk-minimization is a representative quadratic hedging method for contingent claims in incomplete financial markets.
It was undertaken by F\"ollmer and Sondermann \cite{FSo} and F\"ollmer and Schweizer \cite{FS}, and has been well studied for three decades since then.
For an overview of local risk-minimization, see Schweizer \cite{Sch}.
This paper aims to derive explicit representations of locally risk-minimizing (LRM) strategies for financial market models
in which the asset price process is an exponential additive process and to present some numerical results.

There is a large amount of literature on local risk-minimization,
but here we introduce just a part of it by focusing on explicit representations and numerical analysis for market models described by a jump process.
Arai and Suzuki \cite{AS} derived an explicit expression for LRM strategies for models with the asset price process described by
the solution to a L\'evy-driven stochastic differential equation.
In \cite{AS}, they used Malliavin calculus for L\'evy processes by Sol\'e et al. \cite{SUV}, and a Clark-Ocone type formula under a change of measure by Suzuki \cite{Suz}.
Arai et al. \cite{AIS} developed a numerical method for LRM strategies in models with an exponential L\'evy process as the asset price process,
using the Carr-Madan method, which is based on the fast Fourier transform (FFT).
A L\'evy process is, roughly speaking, a stochastic process having independent and stationary increments.
On the other hand, relaxing the requirement of stationary increments in the definition of L\'evy processes, we call such processes additive processes.
That is, the L\'evy measure for an additive process depends on time.
Therefore, it is natural to examine additive processes as a generalization of L\'evy processes,
and deriving expressions and developing numerical methods for market models described by an additive process are significant issues in mathematical finance.

To derive a mathematical expression for LRM strategies in models with an exponential additive asset price process, called exponential additive models henceforth,
it is necessary to employ a Malliavin calculus framework available for additive processes.
In this regard, Malliavin-Skorohod calculus, as developed by Di Nunno and Vives \cite{DV}, is useful,
and a Clark-Ocone-type formula under a change of measure for additive processes obtained by Handa et al. \cite{HSS} is available.
Local risk-minimization has also been discussed in \cite{HSS} using the same sort of argument as \cite{AS};
however, their argument contains inaccurate statements and propositions for which no proof is given.
We should not simply regard an additive process as a non-stationary L\'evy process and think that the only difference lies in whether the L\'evy measure is time-dependent.
We need to note that the time dependence of the L\'evy measure makes a significant difference in the mathematical discussion.

In this paper, we derive expressions of LRM strategies for exponential additive models and emphasize how time dependence affects mathematical discussion.
It is well-known that LRM strategies are given through the F\"ollmer-Schweizer decomposition.
The minimal martingale measure (MMM), one of the equivalent martingale measures, plays a vital role in discussing the F\"ollmer-Schweizer decomposition.
We need the MMM to exist and be square integrable to derive expressions for LRM strategies.
To this end, we will provide conditions that guarantee these properties as integrability conditions on the L\'evy measure,
enabling verification that they are satisfied when a concrete model is given.

As an example of additive processes that frequently appear in mathematical finance, we introduce the variance-gamma scaled self-decomposable (VGSSD) process, 
first proposed by Carr et al. \cite{CGMY}, and conduct numerical experiments of LRM strategies for models written by the VGSSD process.
Here, the VGSSD process is defined as a Sato process (an additive self-similar process with self-decomposable law)
whose distribution at unit time follows a variance-gamma (VG) distribution.
Further, a VG distribution refers to the distribution of a VG process, a L\'evy process defined as a time-changing Brownian motion by a gamma process.

The outline of this paper is as follows. Section 2 introduces fundamental assumptions used throughout the paper and defines LRM strategies and the MMM.
In Section 3, we focus on exponential additive models.
In particular, we provide conditions for discussing LRM strategies and for ensuring the existence and square integrability of the MMM.
Section 4 examines LRM strategies for exponential additive models, the main theme of this paper.
First, we address the result of \cite{HSS} and derive a mathematical expression for LRM strategies in exponential additive models.
Second, we rewrite it as an expression that enables us to compute LRM strategies numerically using the Carr-Madan method.
In Section 5, we introduce the VGSSD process as a representative additive process and conduct numerical experiments.
Section 6 concludes this paper.

%
%
\section{Preliminaries}\setcounter{equation}{0}
Throughout this paper, we consider a financial market model composed of one riskless asset and one risky asset with maturity $T>0$.
Without loss of generality, we may assume that the interest rate is zero.
Suppose that the risky asset price process $S=\{S_t\}_{0\leq t\leq T}$ is given by a positive-valued special semimartingale defined on
a probability space $(\Omega,\calF,\bbP)$ with the canonical decomposition $S_t=S_0+M_t+A_t$,
where $S_0>0$, $M=\{M_t\}_{t\in[0,T]}$ is a martingale with $M_0=0$, and $A=\{A_t\}_{t\in[0,T]}$ is a predictable process of finite variation with $A_0=0$. 
Now, we suppose that $S$ satisfies the following conditions:

\begin{ass}\label{ass-SC}\begin{enumerate}
\item[(SC1)] $\bbE[[M]_T]<\infty$ and $\d{\bbE\l[\l(\int_0^T|dA_t|\r)^2\r]<\infty}$.\vspace{-3mm}
\item[(SC2)] There is a predictable process $\lambda=\{\lambda_t\}_{t\in[0,T]}$ such that $\d{A_t=\int_0^t\lambda_sd\langle M\rangle_s}$. \vspace{-3mm}
\item[(SC3)] The process $K=\{K_t\}_{t\in[0,T]}$ defined as $\d{K_t:=\int_0^t\lambda^2_sd\langle M\rangle_s}$ is continuous such that $K_T<\infty$, $\bbP$-a.s.
\end{enumerate}\end{ass}

\noindent
When $S$ satisfies Assumption \ref{ass-SC} above, $S$ is said to satisfy the structure condition (SC), which is closely related to the no-arbitrage condition.

Here, we define locally risk-minimizing (LRM) strategies as follows:

\begin{defn}\label{defn-LRM}\begin{enumerate}
\item $\Xi=(\xi^1, \xi^2)$ is said to be an $L^2$-strategy if it satisfies the following conditions:\vspace{-2mm}
      \begin{enumerate}
      \item $\xi^1=\{\xi^1_t\}_{t\in[0,T]}$ is a predictable process satisfying
            \begin{equation}\label{eq-xi} \bbE\l[\int_0^T(\xi^1_t)^2d\langle M\rangle_t+\l(\int_0^T|\xi^1_tdA_t|\r)^2\r]<\infty. \end{equation}
            $\xi^1_t$ represents the amount of units of the risky asset an investor holds at time $t$.
      \item $\xi^2=\{\xi^2_t\}_{t\in[0,T]}$ is an adapted process representing the amount of units of the riskless asset.
      \item For any $t\in[0,T]$, denote $V_t(\Xi):=\xi^1_tS_t+\xi^2_t$, which provides the corresponding wealth at time $t\in[0,T]$.
            In particular, $V_0(\Xi)$ gives the initial cost.
            Then, $V(\Xi)=\{V_t(\Xi)\}_{t\in[0,T]}$ is a right continuous process with $\bbE[V_t^2(\Xi)]<\infty$ for every $t\in[0,T]$.
      \end{enumerate}\vspace{-3mm}
\item An $L^2$-strategy $\Xi$ is said to be self-financing, if it satisfies $\d{V_t(\Xi)=V_0(\Xi)+\int_0^t\xi^1_sdS_s}$ for any $t\in[0,T]$.
\item Let $F$ be a square integrable random variable representing the payoff of a contingent claim at maturity $T$.
      For an $L^2$-strategy $\Xi$, the corresponding cost process $C(\Xi)=\{C_t(\Xi)\}_{t\in[0,T]}$ for claim $F$ is defined as
      \[ C_t(\Xi):=F{\bf 1}_{\{t=T\}}+V_t(\Xi)-\int_0^t\xi^1_sdS_s, \ \ \ t\in[0,T]. \]
\item An $L^2$-strategy $\Xi$ is called the LRM strategy for claim $F$, if $V_T(\Xi)=F$, and $[C(\Xi),M]$ is a uniformly integrable martingale.
      Roughly speaking, an $L^2$-strategy $\Xi$, which is not necessarily self-financing, is LRM
      if it is the replicating strategy that minimizes a risk caused by $C(\Xi)$ in the $L^2$-sense among all replicating strategies.
\end{enumerate}\end{defn}

\begin{rem}
Definition \ref{defn-LRM} above is based on Theorem 1.6 of Schweizer \cite{Sch}.
Note that the structure condition (SC) does not usually include the continuity of $K$, but it is included here because we use Theorem 1.6 of \cite{Sch}.
\end{rem}

\noindent
It is sufficient to get a representation of $\xi^1$ in order to obtain a representation of the LRM strategy $\Xi$ for claim $F$,
since $\xi^2$ is automatically determined by $\xi^1$. Here, an $F\in L^2(\bbP)$ admits a F\"ollmer-Schweizer decomposition, if it can be described by
\[ F=F_0+\int_0^T\xi^{FS}_tdS_t+L_T^{FS}, \]
where $F_0\in\bbR$, $\xi^{FS}=\{\xi^{FS}_t\}_{t\in[0,T]}$ is a predictable process satisfying (\ref{eq-xi}),
and $L^{FS}=\{L^{FS}_t\}_{t\in[0,T]}$ is a square-integrable martingale orthogonal to $M$ with $L_0^{FS}=0$.
In addition, Proposition 5.2 of \cite{Sch} proved that, under Assumption \ref{ass-SC},
the LRM strategy $\Xi=(\xi^1,\xi^2)$ for claim $F\in L^2(\bbP)$ exists if and only if $F$ admits a F\"ollmer-Schweizer decomposition;
and $\xi^1_t=\xi^{FS}_t$ holds. As a result, it suffices to obtain a representation of $\xi^F$ in order to get the LRM strategy $\Xi$.

Next, we discuss the minimal martingale measure (MMM). A probability measure $\bbP^*$ equivalent to $\bbP$ is called the MMM if it is an equivalent martingale measure
under which any square integrable $\bbP$-martingale orthogonal to $M$ remains a martingale under $\bbP^*$.
Here, for a semimartingale $Y=\{Y_t\}_{t\in[0,T]}$ with $Y_0=0$, a process $X=\{X_t\}_{t\in[0,T]}$ is called the stochastic exponential of $Y$
if it is the solution to the stochastic differential equation (SDE):
\[ dX_t=X_{t-}dY_t, \ \ \ X_0=1. \]
Wedenote it as $X_t=\calE_t(Y)$.
Now, we define a process $M^\lambda=\{M^\lambda_t\}_{t\in[0,T]}$ as $\d{M^\lambda_t:=-\int_0^t\lambda_sdM_s}$,
and consider its stochastic exponential $Z:=\calE(M^\lambda)$, that is, the process $Z=\{Z_t\}_{t\in[0,T]}$ is given as the solution to the SDE:
\begin{equation}\label{SDE-Z} dZ_t=-\lambda_tZ_{t-}dM_t, \ \ \ Z_0=1. \end{equation}
As seen in Lemma 2.7 of Arai and Suzuki \cite{AS}, when $Z$ is a positive square integrable martingale and Assumption \ref{ass-SC} is satisfied,
the MMM $\bbP^*$ exists with $d\bbP^*=Z_Td\bbP$. Then, for every $t\in[0,T]$, $Z_t$ satisfies
\[ Z_t=\bbE\l[\frac{d\bbP^*}{d\bbP}\Big|\calF_t\r]. \]
In \cite{AS}, for the case where $S$ is given by the solution to an SDE driven by a L\'evy process,
they obtained an explicit representation of the LRM strategy for claim $F$ by using the martingale representation for the product $FZ_T$.
In this paper, we extend \cite{AS}'s result to exponential additive models, but the mathematical treatment is more complicated because the L\'evy measure is no longer stationary,
i.e., time-dependent.

%
%
\section{Exponential additive processes}\setcounter{equation}{0}
In this section, we consider models in which $S=\{S_t\}_{0\leq t\leq T}$ follows an exponential additive process, that is,
the log-price process $L=\{L_t\}_{t\in[0,T]}$ defined as $L_t:=\log(S_t/S_0)$ becomes an additive process.
In particular, we treat only the case where $L$ is an additive process without the Gaussian component.

Let $\{(0,\nu_t,\gamma_t)\}_{t\in[0,T]}$ be the system of generating triplets for the log-price process $L$,
where $\gamma_t$ is a continuous function on $[0,T]$ with $\gamma_0=0$,
and $\{\nu_t\}_{t\in[0,T]}$ is a non-decreasing continuous sequence of L\'evy measures on $\bbR_0$ with $\nu_0\equiv0$.
Here, $\bbR_0:=\bbR\backslash\{0\}$. More precisely, $\{\nu_t\}_{t\in[0,T]}$ satisfies the following conditions:
\begin{enumerate}
\item For each $t\in[0,T]$, $\nu_t$ is a $\sigma$-finite measure such that $\d{\int_{\bbR_0}(1\wedge x^2)\nu_t(dx)<\infty}$.
\item (Non-decreasing property) For any $0\leq s\leq t\leq T$ and any $B\in\calB(\bbR_0)$, $\nu_s(B)\leq\nu_t(B)$, where $\calB(\bbR_0)$ is the Borel algebra on $\bbR_0$.
\item (Continuity) As $s\to t\in[0,T]$, $\nu_s(B)$ converges to $\nu_t(B)$ for any $B\in\calB(\bbR_0)$ with $B\subset\{x:|x|>\ve\}$ for some $\ve>0$.
\end{enumerate}
According to Theorem 9.8, Remark 9.9, and Theorem 11.5 of Sato \cite{S}, there exists an additive process $L$ satisfying the above conditions,
taking a modification as necessary. For every $t\in[0,T]$, $L_t$ has an infinitely divisible law with the characteristic function
\[ \bbE\l[e^{iuL_t}\r]=\exp\l\{iu\gamma_t+\int_{\bbR_0}\l(e^{iux}-1-iux{\bf 1}_{\{|x|\leq1\}}\r)\nu_t(dx)\r\}, \ \ \ u\in\bbR. \]

We denote by $N$ the Poisson random measure associated with $L$, that is, 
\[ N(G):=\#\{t \ | \ (t,\Delta L_t)\in G\}, \ \ \ G\in\calB([0,T]\times\bbR_0), \]
where $\Delta L_t:=L_t-L_{t-}$. Then, its compensated measure $\olnu$ is defined as a measure on $([0,T]\times\bbR_0,\calB([0,T]\times\bbR_0))$ satisfying
$\olnu(G):=\bbE[N(G)]$ for any $G\in\calB([0,T]\times\bbR_0)$. We denote $\tN(G):=N(G)-\olnu(G)$.
Note that $\olnu$ satisfies $\nu_t(A)=\olnu([0,t],A)$ and $\olnu(\{t\},A)=0$ for any $t\in[0,T]$ and any $A\in\calB(\bbR_0)$.
Using $N$ and $\tN$, we can describe $L_t$ as
\[ L_t=\gamma_t+\int_0^t\int_{|x|>1}xN(ds,dx)+\int_0^t\int_{|x|\leq1}x\tN(ds,dx). \]
Here, taking into account Proposition 3.1 of Goutte et al. \cite{GOR}, we define the cumulant generating function $\kappa_t$ as
\begin{equation}\label{eq-kappa} \kappa_t(u):=\log \bbE[e^{uL_t}]=u\gamma_t+\int_0^t\int_{\bbR_0}\l(e^{ux}-1-ux{\bf 1}_{\{|x|\leq1\}}\r)\olnu(ds,dx) \end{equation}
for any $u\in\bbR$ satisfying $\d{\int_0^T\int_{|x|>1}e^{ux}\olnu(ds,dx)<\infty}$.

Now, we assume the following:
\begin{description} 
 \item[(A1)] $\gamma_t$ is differentiable on $(0,T]$, that is, it is expressed as $\d{\gamma_t=\int_0^t\gamma^\prime_sds}$ for any $t\in(0,T]$. 
 \item[(A2)] The L\'evy measure $\olnu$ is absolutely continuous with respect to $dt\times dx$, that is,
       there is a non-negative function $\pi$ such that $\olnu(dt,dx)=\pi(t,x)dtdx$ for any $t\in(0,T]$ and any $x\in\bbR_0$.
 \item[(A3)] $\d{\int_0^T\int_{|x|>1}e^{4x}\olnu(ds,dx)<\infty}$.
\end{description} 

\noindent
We prove one lemma as follows:
\begin{lem}\label{lem1}
Under Assumption (A3), $\d{\int_0^T\int_{\bbR_0}(e^x-1)^2\olnu(dt,dx)<\infty}$ holds.
\end{lem}

\proof
First of all, (A3) implies that
\[ \int_0^T\int_1^\infty(e^x-1)^2\olnu(dt,dx) \leq \int_0^T\int_1^\infty e^{2x}\olnu(dt,dx) \leq \int_0^T\int_1^\infty e^{4x}\olnu(dt,dx) < \infty. \]
Next, since $\d{\int_0^T\int_{\bbR_0}(1\wedge x^2)\olnu(dt,dx) <\infty}$ holds, we have
\[ \int_0^T\int_{-\infty}^{-1}(e^x-1)^2\olnu(dt,dx) \leq \int_0^T\int_{-\infty}^{-1}\olnu(dt,dx) < \infty, \]
\[ \int_0^T\int_0^1(e^x-1)^2\olnu(dt,dx) \leq (e-1)^2\int_0^T\int_0^1x^2\olnu(dt,dx) < \infty, \]
and 
\[ \int_0^T\int_{-1}^0(e^x-1)^2\olnu(dt,dx) \leq \int_0^T\int_{-1}^0x^2\olnu(dt,dx) < \infty. \]
This completes the proof of Lemma \ref{lem1}.
\fin

Under the above setting, the asset price process $S$ is given by the solution to the following SDE:
\begin{equation}\label{SDE-S} dS_t=S_{t-}\l\{\mu^S_tdt+\int_{\bbR_0}\l(e^x-1\r)\tN(dt,dx)\r\}, \ \ \ S_0>0, \end{equation}
where
\begin{equation}\label{eq-muS} \mu^S_t := \gamma^\prime_t+\int_{\bbR_0}\l(e^x-1-x{\bf 1}_{\{|x|\leq1\}}\r)\pi(t,x)dx = \frac{\partial \kappa_t(1)}{\partial t}. \end{equation}
From the view of Lemma \ref{lem1}, we can show that $\mu^S_t$ is well-defined under (A1)--(A3), since $e^x-1-x\leq 2x^2$ holds when $|x|\leq1$.
By the same argument as Proposition 8.20 in Cont and Tankov \cite{CT}, we obtain that $S$, the solution to (\ref{SDE-S}), satisfies $S_t=S_0e^{L_t}$, that is,
\begin{equation}\label{eq-S} S_t = S_0\exp\l\{\gamma_t+\int_0^t\int_{|x|>1}xN(ds,dx)+\int_0^t\int_{|x|\leq1}x\tN(ds,dx)\r\}. \end{equation}

Next, we introduce a sufficient condition under which $S$ satisfies Assumption \ref{ass-SC}. To this end, we denote
\begin{equation}\label{eq-Sigma} \Sigma_t := \int_{\bbR_0}(e^x-1)^2\pi(t,x)dx = \frac{\partial \kappa_t(2)}{\partial t}-2\frac{\partial \kappa_t(1)}{\partial t} \end{equation}
for $t\in(0,T]$, and assume that
\begin{description} 
 \item[(A4)] $0\geq \mu^S_t> -\Sigma_t$ for any $t\in(0,T]$.
\end{description} 
Then, we can see the following:

\begin{prop}\label{prop-SC}
Under Assumptions (A1)--(A4), $S$ satisfies Assumption \ref{ass-SC}.
\end{prop}

\proof
For each $j=1,2$, $X^{(j)}_t:=\exp\{jL_t-\kappa_t(j)\}$ forms a martingale by (A3) and Remark 3.3 in \cite{GOR}. 
Doob's inequalty implies that
\begin{align*}
\bbE\l[\sup_{t\in[0,T]}|S_t|^2\r]
&= S^2_0\bbE\l[\sup_{t\in[0,T]}e^{2L_t}\r] = S^2_0\bbE\l[\sup_{t\in[0,T]}\l|X^{(1)}_te^{\kappa_t(1)}\r|^2\r] \\
&\leq 4S^2_0\sup_{t\in[0,T]}e^{2\kappa_t(1)}\bbE\l[\l|X^{(1)}_T\r|^2\r].
\end{align*}
Note that $\kappa_t(1)$ is a continuous function in $t$ by Proposition 3.6 of \cite{GOR}. Thus, $\sup_{t\in[0,T]}e^{2\kappa_t(1)}$ is finite.
Moreover, we have $\l|X^{(1)}_T\r|^2=X^{(2)}_T\exp\{\kappa_T(2)-2\kappa_T(1)\}$. Thus,  
\[ \bbE\l[\sup_{t\in[0,T]}|S_t|^2\r]\leq4S^2_0\sup_{t\in[0,T]}e^{2\kappa_t(1)}\bbE\l[X^{(2)}_T\r]\exp\{\kappa_T(2)-2\kappa_T(1)\}<\infty. \]
Note that $\kappa_T(2)\in\bbR$ under (A3).

We can see that (SC1) in Assumption \ref{ass-SC} is satisfied by the same sort of argument as Example 2.8 of \cite{AS}.
In fact, Lemma \ref{lem1} and (A4), together with the square integrability of $S$ shown above, ensure that
\begin{align*}
\bbE\l[\l(\int_0^T|dA_t|\r)^2\r] &\leq \bbE\l[\l(\int_0^T S_{t-}|\mu^S_t|dt\r)^2\r] \leq \bbE\l[\sup_{t\in[0,T]}|S_t|^2\r]\l(\int_0^T|\mu^S_t|dt\r)^2 \\
&\leq \bbE\l[\sup_{t\in[0,T]}|S_t|^2\r]\l(\int_0^T\Sigma_tdt\r)^2 <\infty.
\end{align*}
Furthermore, the Burkholder-Davis-Gundy inequality (e.g., see Theorem 48 in Chapter IV in Protter \cite{Protter}) implies that there exists a constant $C$ such that
\[ \bbE[[M]_T] \leq C\bbE\l[\sup_{t\in[0,T]}|M_t|^2\r] \leq C\l\{|S_0|^2+\bbE\l[\sup_{t\in[0,T]}|S_t|^2\r]+\bbE\l[\sup_{t\in[0,T]}|A_t|^2\r]\r\}<\infty. \]

Next, we focus on (SC2) and (SC3). We can see (SC2) by taking $\lambda$ as
\begin{equation}\label{eq-lambda} \lambda_t=\frac{\mu^S_t}{S_{t-}\Sigma_t}, \ \ \ t\in(0,T]. \end{equation}
Note that this is well-defined from the view of Lemma \ref{lem1}. Furthermore, (A4) implies that
\begin{align*}
\int_0^T\lambda^2_td\langle M\rangle_t 
&=    \int_0^T\frac{\l(\mu^S_t\r)^2}{\Sigma_t}dt \leq \int_0^T|\mu^S_t|dt = -\int_0^T\mu^S_tdt\\
&\leq |\gamma_T|+\l|\int_0^T\int_{\bbR_0}\l(e^x-1-x{\bf 1}_{\{|x|\leq1\}}\r)\pi(t,x)dxdt\r| < \infty.
\end{align*}
It is obvious that the process $K$ is continuous. As a result, (SC3) follows.
\fin

\begin{rem}
We can show Lemma \ref{lem1} and Proposition \ref{prop-SC} even if we weaken (A3) to $\d{\int_0^T\int_{|x|>1}e^{2x}\olnu(ds,dx)<\infty}$.\vspace{2mm}
\end{rem}

\vspace{2mm}
Next, we investigate the stochastic process $Z$ defined as the solution to the SDE (\ref{SDE-Z}).
For $t\in[0,T]$ and $x\in\bbR_0$, denote
\[ \theta_{t,x}:=\lambda_tS_{t-}(e^x-1)=\frac{\mu^S_t(e^x-1)}{\Sigma_t}, \]
where $\lambda_t$ is defined in (\ref{eq-lambda}). We can then rewrite the SDE (\ref{SDE-Z}) as $dZ_t = Z_{t-}dM^\lambda_t$, where
\begin{equation}\label{eq-Mlambda}
M^\lambda_t=-\int_0^t\int_{\bbR_0}\theta_{s,x}\tN(ds,dx).
\end{equation}
Here, we show that $Z$ is positive and square integrable. To this end, we add two more assumptions and prepare three lemmas as follows:
\begin{description} 
 \item[(A5)] $\d{\int_0^T\int_{-\infty}^{-1}|\log(1-\theta_{t,x})|^2\olnu(dt,dx)<\infty}$.
 \item[(A6)] $\d{\int_0^T\int_{0<|x|\leq1}|e^x-1|\olnu(ds,dx)<\infty}$, equivalently, $\d{\int_0^T\int_{0<|x|\leq1}|x|\olnu(ds,dx)<\infty}$.
\end{description}

\begin{lem}\label{lem2}
Assuming (A3), we have $\d{\int_0^T\int_{\bbR_0}|e^x-1|^k\olnu(dt,dx)<\infty}$ for $k=3,4$.
\end{lem}

\proof
For $k=3,4$, we have the following: $|e^x-1|^k\leq e^{4x}$ if $x\geq1$, and $|e^x-1|^k\leq1$ if $x\leq-1$.
Moreover, $|e^x-1|^k\leq (e-1)^kx^k\leq(e-1)^4x^2$ if $0<x\leq1$, and $|e^x-1|^k\leq |x|^k\leq x^2$ if $-1\leq x<0$.
Thus, (A3) implies that $\d{\int_0^T\int_{\bbR_0}|e^x-1|^k\olnu(dt,dx)<\infty}$ for $k=3,4$.
\fin

\begin{lem}\label{lem3}
Under Assumptions (A1)--(A5), $\d{\int_0^T\int_{\bbR_0}|\log(1-\theta_{t,x})|^2\olnu(dt,dx)<\infty}$ holds.
\end{lem}

\proof
(A4) implies that $0\leq\log(1-\theta_{t,x})\leq(e^x-1)$ when $x>0$. Thus, we have
\[ \int_0^T\int_0^\infty|\log(1-\theta_{t,x})|^2\olnu(dt,dx) \leq \int_0^T\int_0^\infty(e^x-1)^2\olnu(dt,dx)<\infty \]
by Lemma \ref{lem1}. For $x\in[-1,0)$, we have
\[ 0 \geq \log(1-\theta_{t,x}) \geq \frac{\log(1-\theta_{t,-1})}{\theta_{t,-1}}\theta_{t,x} = \frac{\log(1-\theta_{t,-1})}{e^{-1}-1}(e^x-1), \]
which implies that
\[ \int_0^T\int_{-1}^0|\log(1-\theta_{t,x})|^2\olnu(dt,dx) \leq \int_0^T\int_{-1}^0\l(e^x-1\r)^2\olnu(dt,dx)<\infty. \]
Together with (A5), Lemma \ref{lem3} follows.
\fin

\begin{lem}\label{lem4}
Under Assumptions (A1)--(A6), $\d{\int_0^T\int_{\bbR_0}\l(|\theta_{t,x}|+|\log(1-\theta_{t,x})|\r)\olnu(dt,dx)<\infty}$.
\end{lem}

\proof
First, (A6), together with (A3), ensures that $\d{\int_0^T\int_{\bbR_0}|e^x-1|\olnu(ds,dx)<\infty}$.
Since $|e^x-1|\leq(e-1)|x|$ when $|x|\leq1$, $|\theta_{t,x}|$ is integrable with respect to $\olnu$ by (A4).
By the same way as the proof of Lemma \ref{lem3}, we can see the integrability condition for $|\log(1-\theta_{t,x})|$.
In fact, (A5), together with $\d{\int_0^T\int_{-\infty}^{-1}\olnu(dt,dx)<\infty}$, yields that
\[ \int_0^T\int_{-\infty}^{-1}|\log(1-\theta_{t,x})|\olnu(dt,dx) \leq \int_0^T\int_{-\infty}^{-1}\l(1\vee|\log(1-\theta_{t,x})|^2\r)\olnu(dt,dx)<\infty. \]
\fin

\begin{rem}\label{rem-A5}
\begin{enumerate}
 \item The condition
       \begin{equation}\label{cond-A5} \int_0^T\int_{-\infty}^{-1}\l\{\log\l(1+\frac{\mu^S_t}{\Sigma_t}\r)\r\}^2\olnu(ds,dx)<\infty \end{equation}
       is a sufficient condition for (A5). In fact, when $x<-1$, we have
       \[ 0 \geq \log(1-\theta_{t,x}) \geq \log\l(1+\frac{\mu^S_t}{\Sigma_t}\r), \]
       from which (\ref{cond-A5}) implies that 
       \[ \int_0^T\int_{-\infty}^{-1}|\log(1-\theta_{t,x})|^2\olnu(dt,dx) \leq \int_0^T\int_{-\infty}^{-1}\l\{\log\l(1+\frac{\mu^S_t}{\Sigma_t}\r)\r\}^2\olnu(ds,dx)<\infty. \]
 \item When the log-price process $L$ is a L\'evy process, $\mu^S_t$ and $\Sigma_t$ do not depend on $t$.
       Thus, it is possible to simplify (A5) as
       \begin{equation}\label{cond-AS} \frac{\mu^S_t}{\Sigma_t} > -1+\ve \end{equation}
       for some $\ve>0$, e.g., see Arai et al. \cite{AIS}; however, (\ref{cond-AS}) is too restrictive for the additive process case.
\end{enumerate}
\end{rem}

With the above preparations, we prove the following proposition:

\begin{prop}\label{prop-Z}
Suppose Assumptions (A1)--(A6). Then, $Z$ is a positive square integrable martingale.
\end{prop}

\proof
(A4) ensures that $\theta_{t,x}<1$, that is, $M^\lambda$ defined in (\ref{eq-Mlambda}) satisfies that $\Delta M^\lambda_t>-1$ for any $t\in[0,T]$.
Furthermore, Lemma \ref{lem1} ensures that
\[
\bbE\l[[M^\lambda]_T\r] = \bbE\l[\int_0^T\int_{\bbR_0}\theta_{t,x}^2N(dt,dx)\r] = \int_0^T\int_{\bbR_0}\theta_{t,x}^2\olnu(dt,dx) < \infty,
\]
from which $M^\lambda$ is a uniformly integrable martingale by Corollary 3 of Theorem 26 in Chapter II of \cite{Protter}.
Thus, Theorem 2 of H. Sato \cite{HSato} implies that $Z=\calE(M^\lambda)$ is a positive uniformly integrable martingale, and it is expressed as
\begin{equation}\label{eq-prop-Z-1}
Z_t =\exp\l\{\int_0^t\int_{\bbR_0}\log(1-\theta_{s,x})\tN(ds,dx)+\int_0^t\int_{\bbR_0}\l(\log(1-\theta_{s,x})+\theta_{s,x}\r)\olnu(ds,dx)\r\}.
\end{equation}
Note that the two integrals in (\ref{eq-prop-Z-1}) are well-defined from the view of Lemma \ref{lem4}. We have then
\begin{align*}
Z^2_t
&= \exp\l\{2\int_0^t\int_{\bbR_0}\log(1-\theta_{s,x})\tN(ds,dx)+2\int_0^t\int_{\bbR_0}\l(\log(1-\theta_{s,x})+\theta_{s,x}\r)\olnu(ds,dx)\r\} \\
&= \exp\Bigg\{\int_0^t\int_{\bbR_0}\log(1-2\theta_{s,x}+\theta^2_{s,x})\tN(ds,dx) \\
&\hspace{5mm} +\int_0^t\int_{\bbR_0}\l(\log(1-2\theta_{s,x}+\theta^2_{s,x})+2\theta_{s,x}-\theta^2_{s,x}\r)\olnu(ds,dx)
              +\int_0^t\int_{\bbR_0}\theta^2_{s,x}\olnu(ds,dx)\Bigg\} \\
&= \calE_t\l(\olM^\theta\r)\exp\l\{\int_0^t\int_{\bbR_0}\theta^2_{s,x}\olnu(ds,dx)\r\},
\end{align*}
where $\olM^\theta=\{\olM^\theta_t\}_{t\in[0,T]}$ is the process defined as
\[
\olM^\theta_t:=\int_0^t\int_{\bbR_0}\l(-2\theta_{s,x}+\theta^2_{s,x}\r)\tN(ds,dx), \ \ \ t\in[0,T].
\]
It is enough to see the integrability of $\calE_t\l(\olM^\theta\r)$, since $\d{\int_0^t\int_{\bbR_0}\theta^2_{s,x}\olnu(ds,dx)<\infty}$.
To this end, we calculate the expectation of the quadratic variation of $\olM^\theta$.
From the view of Lemmas \ref{lem1} and \ref{lem2}, we obtain
\begin{align*}
\bbE\l[\l[\olM^\theta\r]_T\r]
&= \bbE\l[\int_0^T\int_{\bbR_0}\l(-2\theta_{t,x}+\theta^2_{t,x}\r)^2N(dt,dx)\r] \\
&= \bbE\l[\int_0^T\int_{\bbR_0}\l(4\theta^2_{t,x}-4\theta^3_{t,x}+\theta^4_{t,x}\r)N(dt,dx)\r] \\
&= \int_0^T\int_{\bbR_0}\l(4\theta^2_{t,x}-4\theta^3_{t,x}+\theta^4_{t,x}\r)\olnu(dt,dx)<\infty,
\end{align*}
which implies that $\olM^\theta$ is a uniformly integrable martingale.
Using Theorem 2 of \cite{HSato} again, we obtain that $\calE\l(\olM^\theta\r)$ is a uniformly integrable martingale.
\fin

\noindent
From the view of the discussion so far, Assumptions (A1)--(A6) implies that the product process $ZS$ is a martingale,
and the MMM $\bbP^*$ exists with the Radon-Nikodym density $\d{\frac{d\bbP^*}{d\bbP}=Z_T}$.

\begin{rem}\label{rem-Situ}
\cite{AS} used Theorem 117 of Situ \cite{Situ} to show the square integrabilities of $S$ and $Z$ for the case where $\olnu$ is stationary.
However, Theorem 117 in \cite{Situ} is not available for the non-stationary case. That is why we take a different approach in this paper.
\end{rem}

%
%
\section{LRM strategies for exponential additive processes}\setcounter{equation}{0}
The objectives of this section are to derive a mathematical expression of LRM strategies for exponential additive models
and evolve it to a numerically tractable form.
Throughout this section, $S$ denotes the exponential additive process defined in (\ref{SDE-S}) or, equivalently, (\ref{eq-S}),
and we consider call options as the claims to be hedged. We denote $F:=(S_T-K)^+$, where $K>0$ is the strike price.

\subsection{On the result of Handa et al. \cite{HSS}} 
Handa et al. \cite{HSS} also discussed local risk-minimization for additive processes.
Note that their discussion is based on the Malliavin-Skorohod calculus undertaken by Di Nunno and Vives \cite{DV}.
In this subsection, we examine Theorem 5.15 of \cite{HSS}.

In the Malliavin-Skorohod calculus, the sample space $\Omega$ must be the canonical space for a pure jump additive process,
of which each element is represented by arranging data from all jumps.
More precisely, $\omega\in\Omega$ is written as $\omega=((t_1,x_1),(t_2,x_2),\dots)$,
where $\omega$ corresponds to the path with jumps of size $x_i$ at time $t_i$ for $i=1,2,\dots$.
Note that $\Omega$ includes all finite and infinite sequences as well as the empty sequence.
Now, we denote $\ve^+_{t,x}\omega:=((t,x),(t_1,,x_1),(t_2,x_2),\dots)$, where $\omega=((t_1,x_1),(t_2,x_2),\dots)$.
That is, $\ve^+_{t,x}\omega$ is the element of $\Omega$ by adding to $\omega$ a jump of size $x$ at time $t$.
Then, for a random variable $X$ on $\Omega$, we define the differential operator $\Psi_{t,x}$ as
\[ \Psi_{t,x}X(\omega):=X(\ve^+_{t,x}\omega)-X(\omega). \]
In this paper, we no longer mention the Malliavin-Skorohod calculus. For its details, see \cite{DV} and \cite{HSS}.

Let us go back to LRM strategies. \cite{HSS} treated financial market models in which the asset price process $S$ is given by $\calE(M^H)$,
where $M^H=\{M^H_t\}_{t\in[0,T]}$ is a process expressed as
\begin{equation}\label{eq-HSS} M^H_t = \int_0^t\alpha_sds+\int_0^t\int_{\bbR_0}\beta_{s,x}\tN(ds,dx). \end{equation}
Here, $\alpha$ and $\beta$ are predictable processes, and $\tN$ is the same one as defined in Section 3.
Furthermore, in their Theorem 5.15, they derived an expression for LRM strategies in this model framework. Now, we modify it to fit our setting.

\begin{prop}\label{prop-HSS}
Suppose the following five conditions: \vspace{-2mm}
\begin{description} 
 \item[(B1)] $\d{\int_0^T\int_{\bbR_0}|\theta_{t,x}|\olnu(dt,dx)<\infty \ \mbox{ and } \ \int_0^T\int_{\bbR_0}|\log(1-\theta_{t,x})|\olnu(dt,dx)<\infty}$. \vspace{-2mm}
 \item[(B2)] $\bbE[Z_TF]<\infty$, $\bbE[Z_T|\Psi_{t,x}F|]<\infty$ for any $t\geq0$ and $x\in\bbR_0$,
             and \\ $\d{\bbE\l[\int_0^T\int_{\bbR_0}|\Psi_{t,x}(Z_TF)|\olnu(dt,dx)\r]<\infty}$. \vspace{-2mm}
 \item[(B3)] The structure condition (SC), that is, Assumption \ref{ass-SC}. \vspace{-2mm}
 \item[(B4)] $Z$ is a positive square integrable martingale; and satisfies the reverse H\"older inequality,
             in other words, we can find a constant $C>0$ such that $\bbE[Z_T^2|\calF_{t-}]\leq CZ_{t-}^2$ for any $t\in(0,T]$. \vspace{-2mm}
 \item[(B5)] $\d{\bbE\l[\int_0^T\int_{\bbR_0}|\Psi_{t,x}F|^2\olnu(dt,dx)\r]<\infty}$. \vspace{-2mm}
\end{description}
Then, the LRM strategy $\xi^{1,F}$ for the call option $F=(S_T-K)^+$ is represented as
\begin{equation}\label{eq-LRM1}
\xi^{1,F}_t=\frac{1}{S_{t-}\Sigma_t}\int_{\bbR_0}\bbE_{\bbP^*}[\Psi_{t,x}F|\calF_{t-}](e^x-1)\pi(t,x)dx,
\end{equation}
where $\d{\Sigma_t:=\int_{\bbR_0}(e^x-1)^2\pi(t,x)dx}$.
\end{prop}

\proof Theorem 5.15 in \cite{HSS} shows that (\ref{eq-LRM1}) holds under their Assumptions 4.2, 5.9, and 5.11, and (5.12);
but note that $\nu_t$ and $\olnu$ are used incorrectly in \cite{HSS}.
Thus, we have only to verify whether the assumptions in \cite{HSS} hold under the five conditions (B1)-(B5).

Remark that $\Psi_{t,x}v=0$ when $v$ is a deterministic function. Moreover,
\[ \Psi_{t,x}(Z_TF)=Z_T\Psi_{t,x}F+F\Psi_{t,x}Z_T+\Psi_{t,x}Z_T\Psi_{t,x}F \]
holds by the definition of $\Psi_{t,x}$. Additionally, $\tH_{t,z}$, which appears in Assumption 4.2 of \cite{HSS}, is given by 1 in our setting.
Thus, Assumption 4.2 of \cite{HSS} is reduced to (B2) and ``$\theta+\log(1-\theta)$, $\log(1-\theta)\in\Dom\Phi$",
where we do not introduce the definition of $\Dom\Phi$, but Remark 5.9 of \cite{DV} ensures that $v\in L^1(\olnu\times\bbP)$ is a sufficient condition for $v\in\Dom\Phi$.
As a result, Assumption 4.2 of \cite{HSS} is derived from (B1) and (B2).
Next, Assumption 5.9 in \cite{HSS} is the same as (B3). In addition, Assumption 5.11 in \cite{HSS} is satisfied by combining (B2) with (B4).

Only (5.12) in \cite{HSS} remains. It can be rewritten as
\[
\bbE\l[\int_0^T\int_{\bbR_0}\l(\bbE_{\bbP^*}[\Psi_{t,x}F|\calF_{t-}]\r)^2\olnu(dt,dx)\r]<\infty
\]
from the view of (5.14) of \cite{HSS} by substituting 1 for $H^*_{t,z}$ in (5.14) of \cite{HSS}.
Thus, (B5), together with the reverse H\"older inequality of $Z$ and Fubini's theorem, implies that
\begin{align*}
\lefteqn{\bbE\l[\int_0^T\int_{\bbR_0}\l(\bbE_{\bbP^*}[\Psi_{t,x}F|\calF_{t-}]\r)^2\olnu(dt,dx)\r]} \\
&= \bbE\l[\int_0^T\int_{\bbR_0}\l(\bbE\l[\frac{Z_T}{Z_{t-}}\Psi_{t,x}F\Big|\calF_{t-}\r]\r)^2\olnu(dt,dx)\r] \\
&\leq \bbE\l[\int_0^T\int_{\bbR_0}\bbE\l[\frac{Z^2_T}{Z^2_{t-}}\Big|\calF_{t-}\r]\bbE\l[|\Psi_{t,x}F|^2\Big|\calF_{t-}\r]\olnu(dt,dx)\r] \\
&\leq C\bbE\l[\int_0^T\int_{\bbR_0}\bbE[|\Psi_{t,x}F|^2|\calF_{t-}]\olnu(dt,dx)\r] \\
&\leq C\bbE\l[\int_0^T\int_{\bbR_0}|\Psi_{t,x}F|^2\olnu(dt,dx)\r]<\infty,
\end{align*}
from which (5.12) in \cite{HSS} follows. This completes the proof of Proposition \ref{prop-HSS}.
\fin

\subsection{Main theorem}
\cite{HSS} asserts in their Corollary 5.17 that an expression for LRM strategies, when $\alpha$ and $\beta$ in (\ref{eq-HSS}) are continuously deterministic functions,
can be derived by the same argument as \cite{AS}; however, as mentioned in Remark \ref{rem-Situ}, the argument in \cite{AS} cannot be applied directly to the case of additive processes.
That is why we employ a different approach to show the square integrability of $Z$ and $S$ in the previous section.
In this subsection, we correct Corollary 5.17 of \cite{HSS} and examine a mathematical expression for LRM strategies in exponential additive models once again.
Therefore, we show that (\ref{eq-LRM1}) holds even if we replace the conditions (B1)-(B5) with (A1)-(A6).
Here, the conditions (A1)-(A6) are mainly the integrability conditions on the L\'evy measure, which are formulated in a useful way to conduct the numerical experiments discussed later.
To simplify notation, the conditions (A1)--(A6) will be collectively referred to Assumption \ref{ass-A} as follows:

\renewcommand{\theassA}{\Alph{assA}}
\begin{assA}\label{ass-A}\begin{description} 
 \item[(A1)] $\gamma_t$ is differentiable, that is, it is expressed as $\d{\gamma_t=\int_0^t\gamma^\prime_sds}$. \vspace{-2mm}
 \item[(A2)] The L\'evy measure $\olnu$ is absolutely continuous with respect to $dt\times dx$, that is,
       there is a non-negative function $\pi$ such that $\olnu(dt,dx)=\pi(t,x)dtdx$ for any $t\in(0,T]$ and any $x\in\bbR_0$. \vspace{-5mm}
 \item[(A3)] $\d{\int_0^T\int_{|x|>1}e^{4x}\olnu(ds,dx)<\infty}$. \vspace{-2mm}
 \item[(A4)] $0\geq \mu^S_t> -\Sigma_t$ for any $t\in(0,T]$, where $\mu^S_t$ and $\Sigma_t$ are defined in (\ref{eq-muS}) and (\ref{eq-Sigma}), respectively.
 \item[(A5)] $\d{\int_0^T\int_{-\infty}^{-1}|\log(1-\theta_{t,x})|^2\olnu(dt,dx)<\infty}$. Recall that $\d{\theta_{t,x}:=\frac{\mu^S_t(e^x-1)}{\Sigma_t}}$. \vspace{-2mm}
 \item[(A6)] $\d{\int_0^T\int_{0<|x|\leq1}|e^x-1|\olnu(ds,dx)<\infty}$, equivalently, $\d{\int_0^T\int_{0<|x|\leq1}|x|\olnu(ds,dx)<\infty}$.
\end{description}\end{assA}

\begin{thm}\label{thm-A} (\ref{eq-LRM1}) holds under Assumption \ref{ass-A}. \end{thm}

\proof
It suffices to see that all of (B1)--(B5) hold under Assumption \ref{ass-A}.
Lemma \ref{lem4} and  Proposition \ref{prop-Z} guarantee (B1) and the first condition of (B4), respectively. As seen in Proposition \ref{prop-SC}, (B3) follows. 
In addition, $\d{\sup_{t\in[0,T]}|S_t|}$ is square integrable by the proof of Proposition \ref{prop-SC}.
Hence, we obtain the first condition of (B2). A similar calculation to Lemmas 5.21 and 5.22 of \cite{HSS} provides
\begin{equation}\label{eq-PsiF} \Psi_{t,x}F = \Psi_{t,x}(S_T-K)^+ = (e^xS_T-K)^+-(S_T-K)^+. \end{equation}
Thus, $|\Psi_{t,x}F| \leq (e^x+1)S_T$, which implies the second condition of (B2).
Next, we examine the second condition of (B4). 
First, $Z=\calE(M^\lambda)$ and $\langle M^\lambda\rangle=K$ hold, where $M^\lambda$ is defined in (\ref{eq-Mlambda}) and $K$ is defined in (SC3) of Assumption \ref{ass-SC}.
Note that $K_T$ is finite and deterministic, that is, in $L^\infty$; and $M^\lambda$ is a square integrable martingale by the proof of Proposition \ref{prop-Z}.
Thus, Proposition 3.7 of Choulli et al. \cite{CKS} implies that $Z$ satisfies the reverse H\"older inequality.

Now, we present (B5) before proving the third condition of (B2). By (\ref{eq-PsiF}), we have
\begin{align}\label{eq-thm-A-1}
(\Psi_{t,x}F)^2
&= \l\{(e^xS_T-K)^+-(S_T-K)^+\r\}^2 \nonumber \\
&= (e^x-1)^2S_T^2{\bf 1}_{\{S_T\geq K\}}{\bf 1}_{\{x\geq\log(K/S_T)\}}+(S_T-K)^2{\bf 1}_{\{S_T\geq K\}}{\bf 1}_{\{x<\log(K/S_T)\}} \nonumber \\
&\hspace{5mm} +(e^xS_T-K)^2{\bf 1}_{\{S_T<K\}}{\bf 1}_{\{x\geq\log(K/S_T)\}} \nonumber \\
&\leq (e^x-1)^2S_T^2{\bf 1}_{\{S_T\geq K\}}{\bf 1}_{\{x\geq\log(K/S_T)\}}+(S_T-e^xS_T)^2{\bf 1}_{\{S_T\geq K\}}{\bf 1}_{\{x<\log(K/S_T)\}} \nonumber \\
&\hspace{5mm} +(e^xS_T-S_T)^2{\bf 1}_{\{S_T<K\}}{\bf 1}_{\{x\geq\log(K/S_T)\}} \nonumber \\
&\leq (e^x-1)^2S_T^2.
\end{align}
By (A3), Lemma \ref{lem1}, and the square integrability of $S_T$, (B5) is satisfied.

It remains to show the third condition of (B2). To this end, we calculate $\Psi_{t,x}(Z_TF)$.
By the definition of $\Psi_{t,x}$, we have $\Psi_{t,x}Z_T=-\theta_{t,x}Z_T$. Thus, (\ref{eq-PsiF}) implies that
\begin{align*}
\Psi_{t,x}(Z_TF)
&= Z_T\Psi_{t,x}F+F\Psi_{t,x}Z_T+\Psi_{t,x}F\cdot\Psi_{t,x}Z_T \\
&= Z_T\Psi_{t,x}F-\theta_{t,x}Z_T(F+\Psi_{t,x}F) = Z_T\Psi_{t,x}F-\theta_{t,x}Z_T(e^xS_T-K)^+.
\end{align*}
We have then
\begin{align*}
|\Psi_{t,x}(Z_TF)|
&\leq Z_T\l\{|\Psi_{t,x}F|+|\theta_{t,x}|(e^xS_T-K)^+\r\} \leq Z_T|e^x-1|\l(S_T+\frac{|\mu^S_t|}{\Sigma_t}e^xS_T\r) \\
&\leq Z_TS_T|e^x-1|(1+e^x) \leq Z_TS_T\l\{(1+e)|e^x-1|{\bf 1}_{\{x\leq1\}}+e^{2x}{\bf 1}_{\{x>1\}}\r\}
\end{align*}
by (\ref{eq-thm-A-1}). Consequently, (A3) and (A6) imply the third condition of (B2).
\fin

\subsection{Numerically tractable expression}
The goal of this subsection is to convert (\ref{eq-LRM1}) into a numerically tractable form.
Throughout this subsection, we fix $t\in[0,T]$ and $K>0$. By (\ref{eq-PsiF}), (\ref{eq-LRM1}) is rewritten as
\begin{equation}\label{eq-LRM2} \xi^{1,F}_t=\frac{1}{S_{t-}\Sigma_t}\int_{\bbR_0}\bbE_{\bbP^*}[(e^xS_T-K)^+-(S_T-K)^+|\calF_{t-}](e^x-1)\pi(t,x)dx. \end{equation}
Denote the characteristic function of $L_T-L_t=\log(S_T/S_t)$ under the MMM $\bbP^*$ as
\[ \phi^*_{t,T}(z):=\bbE_{\bbP^*}\l[e^{iz(L_T-L_t)}\r], \ \ \ z\in\bbC. \]
For later use, we calculate $\phi^*_{t,T}$. To this end, we rewrite (\ref{SDE-S}) as
\begin{equation}\label{SDE-S*} dS_t=S_{t-}\int_{\bbR_0}\l(e^x-1\r)\tN^*(dt,dx), \ \ \ S_0>0, \end{equation}
where $\pi^*(t,x):=\l(1-\theta_{t,x}\r)\pi(t,x)$ and $\tN^*(dt,dx):=N(dt,dx)-\pi^*(t,x)dtdx$.
Solving the SDE (\ref{SDE-S*}) by the same way as (\ref{eq-S}), we obtain
\[ L_t=\int_0^t\int_{\bbR_0}\l(1+x{\bf 1}_{\{|x|\leq1\}}-e^x\r)\pi^*(s,x)dsdx+\int_0^t\int_{|x|>1}xN(ds,dx)+\int_0^t\int_{|x|\leq1}x\tN^*(ds,dx), \]
which implies that
\begin{equation}\label{eq-phi*} \phi^*_{t,T}(z) = \exp\l\{\int_t^T\int_{\bbR_0}\l(e^{izx}-ize^x+iz-1\r)\pi^*(s,x)dxds\r\}. \end{equation}

We can see the following proposition:
\begin{prop}[Carr-Madan method]\label{prop-T}
Under Assumption \ref{ass-A}, we have
\begin{equation}\label{eq-prop-T} \bbE_{\bbP^*}[(S_T-K)^+|\calF_{t-}]
=\frac{1}{\pi}\int_0^\infty K^{-iu-R+1}\frac{\phi^*_{t,T}(u-iR)S_{t-}^{iu+R}}{(iu+R-1)(iu+R)}du \end{equation}
for any $R\in(1,2]$. Note that the right-hand side is independent of the choice of $R\in(1,2]$.
\end{prop} 

\proof Since $\Delta L_t=0$, $\bbP$-a.s., $L_T-L_t$ has the same law under $\bbP^*$ as $L_T-L_{t-}$. Thus, we have
\[ \bbE_{\bbP^*}[(S_T-K)^+|\calF_{t-}] = \bbE_{\bbP^*}\l[\l(S_0e^{L_T-L_t+L_{t-}}-K\r)^+ \Big| \calF_{t-}\r] = \bbE_{\bbP^*}\l[f(X_{t,T}-\log(K/S_{t-})) | \calF_{t-}\r], \]
where $X_{t,T}:=L_T-L_t$ and $f(x):=K(e^x-1)^+$. Fix $R\in(1,2]$ arbitrarily, and denote $g(x):=e^{-R x}f(x)$ and $\whg(u):=\d{\int_{\bbR}e^{iux}g(x)dx}$.
Then, we can see that the following three conditions hold: \vspace{-2mm}
\begin{description} 
 \item[(a)] $g$ is a bounded continuous function in $L^1(\bbR)$.\vspace{-2mm}
 \item[(b)] $\bbE_{\bbP^*}\l[e^{R X_{t,T}}\r]<\infty$.\vspace{-2mm}
 \item[(c)] $\whg\in L^1(\bbR)$.\vspace{-2mm}
\end{description} 
In fact, Condition (a) holds clearly, Lemma \ref{lem-T} below ensures Condition (b), and Condition (c) is satisfied, since we have
\begin{align*}
\whg(u) &= K\int_{\bbR}e^{iux}e^{-Rx}(e^x-1)^+dx = K\int_0^\infty \l(e^{(iu-R+1)x}-e^{(iu-R)x}\r)dx \\
        &= K\l(-\frac{1}{iu-R+1}+\frac{1}{iu-R}\r) = \frac{K}{(iu-R+1)(iu-R)},
\end{align*}
and there is no real root for $(iu-R+1)(iu-R)=0$. Denote
\[ \whf(-u+iR) := \int_{\bbR}e^{(-iu-R)x}f(x)dx = \whg(-u). \]
Then, Theorem 2.2 of Eberlein et al. \cite{EGP} implies that
\begin{align*}
\bbE_{\bbP^*}[(S_T-K)^+|\calF_{t-}]
&= \frac{1}{2\pi}\int_{\bbR}e^{(-R-iu)\log(K/S_{t-})}\phi^*_{t,T}(u-iR)\whf(-u+iR)du \\
&= \frac{1}{2\pi}\int_{\bbR}e^{(-R-iu)\log(K/S_{t-})}\frac{\phi^*_{t,T}(u-iR)K}{(-iu-R+1)(-iu-R)}du \\
&= \frac{1}{2\pi}\int_{\bbR}K^{-iu-R+1}\frac{\phi^*_{t,T}(u-iR)S_{t-}^{iu+R}}{(iu+R-1)(iu+R)}du \\
&= \frac{1}{\pi}\int_0^\infty K^{-iu-R+1}\frac{\phi^*_{t,T}(u-iR)S_{t-}^{iu+R}}{(iu+R-1)(iu+R)}du.
\end{align*}
Note that the last equality holds because the integral is real-valued and the real part of the integrand is even.
Moreover, its value is independent of the choice of $R\in(1,2]$. Consequently, Proposition \ref{prop-T} follows.
\fin

\begin{lem}\label{lem-T} $\bbE_{\bbP^*}\l[e^{2X_{t,T}}\r]<\infty$, where $X_{t,T}=L_T-L_t$. \end{lem}

\proof
We define the process $M^*=\{M^*_t\}_{t\in[0,T]}$ as
\[ M^*_t=\int_0^t\int_{\bbR_0}\l(e^x-1\r)\tN^*(ds,dx). \]
We have then $S=\calE(M^*)$. Here, we show that $M^*$ is a square integrable $\bbP^*$-martingale such that $\langle M^* \rangle$ is bounded.
To this end, it suffices to make sure that $\d{\int_0^T\int_{\bbR_0}\l(e^x-1\r)^2\pi^*(t,x)dxdt<\infty}$;
however, we can derive this from Lemmas \ref{lem1} and \ref{lem2} together with (A4).
By Proposition 3.7 of \cite{CKS}, $S$ satisfies the reverse H\"older inequality under $\bbP^*$, that is,
there is a constant $C>0$ such that $\d{\bbE_{\bbP^*}[S_T^2|\calF_t]\leq CS_t^2}$. 
As a result, we have
\[ \bbE_{\bbP^*}\l[e^{2X_{t,T}}\r] = \bbE_{\bbP^*}\l[e^{2(L_T-L_t)}\r] = \bbE_{\bbP^*}\l[\l(\frac{S_T}{S_t}\r)^2\r]
= \bbE_{\bbP^*}\l[\frac{1}{S_t^2}\bbE_{\bbP^*}[S_T^2|\calF_t]\r]\leq C. \]  
\fin

Let us further modify (\ref{eq-LRM2}). First, we regard the call option price as a function of $K$, and denote it by
\[ f(K):=\bbE_{\bbP^*}[(S_T-K)^+|\calF_{t-}], \ \ \ K>0. \]
Proposition \ref{prop-T} provides that
\[ f(K) = \frac{1}{\pi}\int_0^\infty K^{-iu-R+1}\frac{\phi^*_{t,T}(u-iR)S_{t-}^{iu+R}}{(iu+R-1)(iu+R)}du. \]
Then, since $\bbE_{\bbP^*}[(e^xS_T-K)^+|\calF_{t-}] = e^xf(e^{-x}K)$, we obtain
\begin{align}\label{eq-CM1}
\xi^{1,F}_t
&= \frac{1}{S_{t-}\Sigma_t}\int_{\bbR_0}\{e^xf(e^{-x}K)-f(K)\}(e^x-1)\pi(t,x)dx \nonumber \\
&= \frac{1}{S_{t-}\Sigma_t}\int_{\bbR_0}\frac{1}{\pi}\int_0^\infty\l\{e^{(iu+R)x}-1\r\}K^{-iu-R+1}
   \frac{\phi^*_{t,T}(u-iR)S_{t-}^{iu+R}}{(iu+R-1)(iu+R)}du(e^x-1)\pi(t,x)dx \nonumber \\
&= \frac{1}{\pi S_{t-}\Sigma_t}\int_0^\infty K^{-iu-R+1}\int_{\bbR_0}(e^{(iu+R)x}-1)(e^x-1)\pi(t,x)dx
  \frac{\phi^*_{t,T}(u-iR)S_{t-}^{iu+R}}{(iu+R-1)(iu+R)}du.
\end{align}
If we can obtain an explicit expression for the inner integral $\d{\int_{\bbR_0}(e^{(iu+R)x}-1)(e^x-1)\pi(t,x)dx}$,
then (\ref{eq-CM1}) could be computed numerically using the fast Fourier transform (FFT).

Now, we focus on the inner integral. To this end, we extend the domain of $\kappa_t$ defined in (\ref{eq-kappa}) to 
\[ \calD:=\l\{z\in\bbC \ \Big| \ \int_0^T\int_{|x|>1}e^{\Re(z)x}\pi(t,x)dxdt < \infty\r\}, \]
where $\Re(z)$ denotes the real part of $z$. Here, recall that
\[ \kappa_t(z):=z\gamma_t+\int_0^t\int_{\bbR_0}\l(e^{zx}-1-zx{\bf 1}_{\{|x|\leq1\}}\r)\pi(s,x)dxds. \]
Furthermore, we denote $\d{l_t(z):=\frac{\partial \kappa_t(z)}{\partial t}}$ for $z\in\calD$.
Then, we can see the following:

\begin{lem}\label{lem-Psi}
For any $z\in\bbC$ with $z+1\in\calD$, we have
\[ \int_{-\infty}^\infty \l(e^{zx}-1\r)\l(e^x-1\r)\pi(t,x)dx = l_t(z+1)-l_t(z)-l_t(1). \]
\end{lem}

\proof
Since $l_t(z)=z\gamma^\prime_t+\d{\int_{\bbR_0}\l(e^{zx}-1-zx{\bf 1}_{\{|x|\leq1\}}\r)\pi(t,x)dx}$ for $z\in\calD$, we obtain
\begin{align*}
\lefteqn{l_t(z+1)-l_t(z)-l_t(1)} \\
&= (z+1)\gamma^\prime_t-z\gamma^\prime_t-\gamma^\prime_t \\
&\hspace{5mm}+\int_{\bbR_0}\l(e^{(z+1)x}-1-(z+1)x{\bf 1}_{\{|x|\leq1\}}-e^{zx}+1+zx{\bf 1}_{\{|x|\leq1\}}-e^{x}+1+x{\bf 1}_{\{|x|\leq1\}}\r)\pi(t,x)dx \\
&= \int_{\bbR_0}\l(e^{(z+1)x}-e^{zx}-e^x+1\r)\pi(t,x)dx.
\end{align*}
\fin

\noindent
Since $iu+R+1\in\calD$ for any $R\in(1,2]$ and $u\in\bbR$ by (A3), Lemma \ref{lem-Psi} implies that 
\begin{equation}\label{eq-InnerInt} \int_{\bbR_0}(e^{(iu+R)x}-1)(e^x-1)\pi(t,x)dx = l_t(iu+R+1)-l_t(iu+R)-l_t(1). \end{equation}
In addition, we calculate the characteristic function $\phi^*_{t,T}$ using Lemma \ref{lem-Psi}. (\ref{eq-phi*}) implies that
\begin{align}\label{eq-phi*2}
\log(\phi^*_{t,T}(z)) 
&= \int_t^T\int_{\bbR_0}\l(e^{izx}-ize^x+iz-1\r)\pi^*(s,x)dxds \nonumber \\
&= \int_t^T\int_{\bbR_0}\l(e^{izx}-1-iz(e^x-1)\r)\l(1-\frac{\mu^S_t}{\Sigma_t}(e^x-1)\r)\pi(s,x)dxds \nonumber \\
&= \int_t^T\l(l_s(iz)-izl_s(1)\r)ds-\int_t^T\frac{\mu^S_t}{\Sigma_t}\int_{\bbR_0}\l(e^{izx}-1-iz(e^x-1)\r)(e^x-1)\pi(s,x)dxds \nonumber \\
&= \int_t^T\l(l_s(iz)-izl_s(1)\r)ds-\int_t^T\frac{\mu^S_t}{\Sigma_t}\l\{l_s(iz+1)-l_s(iz)-l_s(1)-iz\l(l_s(2)-2l_s(1)\r)\r\}ds.
\end{align}

%
%
\section{Numerical experiments for VGSSD}\setcounter{equation}{0}
In this section, we introduce the variance-gamma scaled self-decomposable (VGSSD) process, undertaken by Carr et al. \cite{CGMY}, as an example of additive processes
and, using (\ref{eq-prop-T}), (\ref{eq-InnerInt}), and (\ref{eq-phi*2}), conduct numerical experiments computing the values of LRM strategies for exponential VGSSD models.

To introduce the VGSSD process, we prepare some terminologies.
\begin{defn}
\begin{enumerate}
\item A random variable $X$ is said to be self-decomposable if, for any $c\in(0,1)$, there exists an independent random variable $X^{(c)}$ such that $X\dequiv cX+X^{(c)}$,
      where $X_1\dequiv X_2$ means that $X_1$ and $X_2$ are identically distributed.
\item A stochastic process $Y=\{Y_t\}_{t\geq0}$ is self-similar if there is an $H>0$ such that $Y_{at}\dequiv a^HY_t$ for any $a>0$ and $t\geq0$.
      To emphasize the exponent $H$, we say that $Y$ is $H$-self-similar.
      Furthermore, $Y$ is called broad-sense $H$-self-similar if, for any $a>0$, there is a function $c_t$ from $[0,\infty)$ to $\bbR$ such that
      $Y_{at}\dequiv a^HY_t+c_t$ holds for any $t\geq0$.
\item A stochastic process is called a Sato process if it is a self-similar additive process, and the distribution at unit time is self-decomposable.
\end{enumerate}
\end{defn}

\noindent
The VGSSD process is a generalization of the variance-gamma (VG) process.
Here, a L\'evy process defined as a Brownian motion subordinated to a gamma process is called a VG process.
In other words, VG processes can be described as $\sigma B_{G_t}+\mu G_t$, where $\sigma>0$, $\mu\in\bbR$, $\{B_t\}_{t\geq0}$ is a one-dimensional standard Brownin motion,
and $\{G_t\}_{t\geq0}$ is a Gamma process.
Besides, any VG process does not possess a Brownian component, and its L\'evy density $\nu^{VG}$ is given by $\d{\nu^{VG}(x)=\frac{h^{VG}(x)}{|x|}}$ for $x\in\bbR_0$, where
\[ h^{VG}(x)=C\l(e^{-Mx}{\bf 1}_{\{x>0\}}+e^{Gx}{\bf 1}_{\{x<0\}}\r), \ \ \ C, G, M>0. \]
Note that the distribution of a VG process at unit time is called a VG distribution, which is self-decomposable; and
a VGSSD process is defined as a Sato process whose distribution at unit time follows a VG distribution.
Here, let $X^{VG}=\{X^{VG}_t\}_{t\geq0}$ denote a VG process with parameters $C$, $G$, and $M$.
We consider a market model in which the log-price process $L$ is given by a VGSSD process associated with $X^{VG}$, with self-similarity exponent $H$.
Note that $L$ is $H$-self-similar, and $L_1\dequiv X^{VG}_1$ holds.
Then, Theorem 1 of \cite{CGMY} implies that the L\'evy density $\pi$ and the drift function $\gamma_t$ for $L$ are, respectively, given by
\[ \pi(t,x)=CHt^{-1-H}\l(Me^{-Mt^{-H}x}{\bf 1}_{\{x>0\}}+Ge^{Gt^{-H}x}{\bf 1}_{\{x<0\}}\r), \]
and
\[ \gamma_t = \int_0^t\int_{|x|\leq1}x\olnu(ds,dx). \]

Next, we add a differentiable function $\rho_t$ to $\gamma_t$ as follows:
\[ \gamma_t = \rho_t+\int_0^t\int_{|x|\leq1}x\olnu(ds,dx), \ \ \ t\in(0,T]. \]
Under this setting, $L$ is no longer self-similar, but it still retains the broad-sense self-similarity.
Now, we examine Assumption \ref{ass-A}. First, (A1) and (A2) are satisfied. Additionally, suppose
\begin{equation}\label{ass-VGSSD} 4<MT^{-H}. \end{equation}
Then, (A3) is also satisfied. Besides, we can confirm (A6) as follows: 
\begin{align*}
\int_0^T\int_{|x|\leq1}|x|\pi(t,x)dxdt
&=    CH\int_0^Tt^{-1-H}\l\{\int_0^1Mxe^{-Mt^{-H}x}dx+\int_{-1}^0G(-x)e^{Gt^{-H}x}dx\r\}dt \\
&\leq CH\int_0^Tt^{-1-H}\l\{\frac{t^{2H}}{M}+\frac{t^{2H}}{G}\r\}dt < \infty.
\end{align*}
To discuss (A4), we denote
\[ q_t(z) := \int_{\bbR_0}\l(e^{zx}-1\r)\pi(t,x)dx \]
for $z\in\calD$. We have then
\begin{align*}
q_t(z)
&= CHt^{-1-H}\l\{\int_0^\infty\l(e^{zx}-1\r)Me^{-Mt^{-H}x}dx+\int_{-\infty}^0\l(e^{zx}-1\r)Ge^{Gt^{-H}x}dx\r\} \\
&= CHt^{-1-H}\l\{M\int_0^\infty\l(e^{-(Mt^{-H}-z)x}-e^{-Mt^{-H}x}\r)dx+G\int_{-\infty}^0\l(e^{(Gt^{-H}+z)x}-e^{Gt^{-H}x}\r)dx\r\} \\
&= CHt^{-1-H}\l\{M\l(\frac{1}{Mt^{-H}-z}-\frac{1}{Mt^{-H}}\r)+G\l(\frac{1}{Gt^{-H}+z}-\frac{1}{Gt^{-H}}\r)\r\} \\
&= CHt^{-1}z\l\{\frac{1}{Mt^{-H}-z}-\frac{1}{Gt^{-H}+z}\r\} = CHt^{-1+H}z\frac{G-M+2zt^H}{(M-zt^H)(G+zt^H)}.
\end{align*}
Thus, $l_t(z)$ is given as
\begin{equation}\label{eq-func_l} l_t(z) = z\gamma^\prime_t+\int_{\bbR_0}\l(e^{zx}-1-zx{\bf 1}_{\{|x|\leq1\}}\r)\pi(t,x)dx = z\rho^\prime_t+q_t(z), \end{equation}
from which we obtain $\mu^S_t = l_t(1) = \rho^\prime_t+q_t(1)$ and $\Sigma_t = l_t(2)-2l_t(1) = q_t(2)-2q_t(1)$.
As a result, (A4), i.e., $0\geq \mu^S_t> -\Sigma_t$ holds for any $t\in(0,T]$, is equivalent to
\begin{equation}\label{ass-VGSSD-2} -q_t(1) \geq \rho^\prime_t > q_t(1)-q_t(2) \end{equation}
for any $t\in(0,T]$.
As for (A5), as the derivation of sufficient conditions is complicated, it shall be examined individually for given concrete examples.

We introduce two examples that satisfy (\ref{ass-VGSSD-2}) and (A5) as long as $M$ and $T$ are set to meet (\ref{ass-VGSSD}).
\begin{ex}\label{ex-VGSSD}
\begin{enumerate}
\item[(1)] The first example is when $\rho^\prime_t=-q_t(1)$, corresponding to the case where $\mu^S_t=\theta_{t,x}=0$ for any $t\in(0,T]$ and any $x\in\bbR_0$,
      that is, the asset price process $S$ becomes a $\bbP$-martingale. Thus, $\bbP^*=\bbP$ holds. In this case, (A5) is automatically satisfied.
      Besides, (\ref{ass-VGSSD-2}) is also satisfied, since $q_t(2)>2q_t(1)$ holds. \vspace{-3mm}
\item[(2)] Next is the case where $\d{\frac{\mu^S_t}{\Sigma_t}=-\frac{1}{2}}$, in other words, $\d{\rho^\prime_t=-\frac{1}{2}q_t(2)}$.
     We can see immediately that (\ref{ass-VGSSD-2}) and (A5) are satisfied.
\end{enumerate}
\end{ex}

\noindent
Here, we conduct numerical experiments for the two models introduced in Example \ref{ex-VGSSD}.
To this end, taking into account (\ref{eq-InnerInt}), we discretize the right-hand side of (\ref{eq-CM1}) as
\begin{align}\label{eq-approx}
\lefteqn{\frac{1}{\pi S_{t-}\Sigma_t}\sum_{j=0}^{N-1}e^{(-iz_j+1)k}\l\{l_t(iz_j+1)-l_t(iz_j)-l_t(1)\r\}\frac{\phi^*_{t,T}(z_j)S_{t-}^{iz_j}}{(iz_j-1)iz_j}\eta} \nonumber \\
&\hspace{8mm} = \frac{1}{\pi(q_t(2)-2q_t(1))}\sum_{j=0}^{N-1} e^{(-iz_j+1)k}\l\{q_t(iz_j+1)-q_t(iz_j)-q_t(1)\r\}\frac{\phi^*_{t,T}(z_j)S_{t-}^{iz_j-1}}{(iz_j-1)iz_j}\eta,
\end{align}
where $N\in\bbN$ is the number of grid points, $\eta>0$ is the distance between adjacent grid points, $k=\log K$ and $z_j:=\eta j-iR$ for $j=0,\dots,N-1$. 
With appropriate settings of $N$ and $\eta$, (\ref{eq-approx}) provides a sufficiently accurate approximation for $\xi^{1,F}_t$.
To compute (\ref{eq-approx}), we employ the FFT. In our experiments, we take $N=2^{14}$, $\eta=0.25$, and $R=1.75$, which are the standard settings for the FFT.
Throughout all experiments, we set $C=1$, $G=M$, $H=\d{\frac{1}{2}}$, and $S_{t-}=1$. Then, we obtain
\[ q_t(z) = \frac{z^2}{M^2-z^2t}, \ \ \ \mbox{ and } \ \ \ \int_t^Tq_s(z)ds = \log\l(\frac{M^2-z^2t}{M^2-z^2T}\r). \]
From the view of (\ref{eq-phi*2}) and (\ref{eq-func_l}), $\phi^*_{t,T}$ for Example \ref{ex-VGSSD} (1) satisfies
\[ \log(\phi^*_{t,T}(z)) = \int_t^T\l(q_s(iz)-izq_s(1)\r)ds = \log\l(\frac{M^2+z^2t}{M^2+z^2T}\r)-iz\log\l(\frac{M^2-t}{M^2-T}\r). \]
As for (2), we have the following:
\begin{align*}
\log(\phi^*_{t,T}(z))
&= \int_t^T\l(q_s(iz)-izq_s(1)\r)ds+\frac{1}{2}\int_t^T\l\{q_s(iz+1)-q_s(iz)-q_s(1)-iz\l(q_s(2)-2q_s(1)\r)\r\} \\
&= \int_t^T\l\{\frac{1}{2}q_s(iz+1)+\frac{1}{2}q_s(iz)-\frac{1}{2}q_s(1)-izq_s(2)\r\}ds \\
&= \frac{1}{2}\log\l(\frac{M^2-(iz+1)^2t}{M^2-(iz+1)^2T}\r)+\frac{1}{2}\log\l(\frac{M^2+z^2t}{M^2+z^2T}\r)-\frac{1}{2}\log\l(\frac{M^2-t}{M^2-T}\r) \\
&\hspace{5mm} -iz\log\l(\frac{M^2-4t}{M^2-4T}\r).
\end{align*}

We implement two types of experiments as follows:
\begin{enumerate}
\item[(A)] Fix $T=1$ and $K=1$, i.e., we treat the at-the-money (ATM) options.
 Instead, we vary $t$ from $0.01$ to $0.99$ in steps of $0.01$, that is, we compute LRM strategies for ATM call options for various time-to-maturities.
 Note that, although the value of $t$ varies, we set $S_t$ to $1$ for any $t$.
\item[(B)] Fix $t=0$ and $T=0.5$, and vary $K$ from $0.51$ to $1.50$ in steps of $0.01$.
 That is, we compute LRM strategies at time to maturity of $0.5$ for various moneyness levels, from out-of-the-money (OTM) to in-the-money (ITM). 
\end{enumerate}
We conduct the types (A) and (B) of numerical experiments for the models (1) and (2) in Example \ref{ex-VGSSD}, that is, a total of four experiments.
In all experiments, we perform computations in two settings: $M=4$ and $M=16$. All experimental results are shown in Figure \ref{fig1}.
Panels (A1) and (A2) in Figure \ref{fig1} draw the values of LRM strategies for the ATM call option with maturity $T=1$ at $t=0.01,0.02,\dots,0.99$,
where we set $S_t$ to $1$ for all $t=0.01,0.02,\dots,0.99$.
On the other hand, Panels (B1) and (B2) represent the values of LRM strategies at time $0$ for call options with maturity $0.5$ vs. strike prices $K=0.51,0.52,\dots,1.50$,
where the present asset price $S_0$ is $1$. In each panel, the red and blue curves correspond to $M=4$ and $M=16$, respectively.
In panels (B1) and (B2), the blue curves ($M=16$) take values near $1$ when $K<1$ (i.e., ITM options), change rapidly around the ATM, and become almost $0$ for OTM options.
On the other hand, the red curves ($M=4$) change less rapidly than the blue ones.
This result is consistent with the fact that the larger the value of $M$, the less likely large jumps are to occur, i.e., the smaller the fluctuations in asset prices.

\begin{figure}[H]
 \renewcommand{\thesubfigure}{A1} \begin{minipage}{0.5\hsize}\includegraphics[width=70mm]{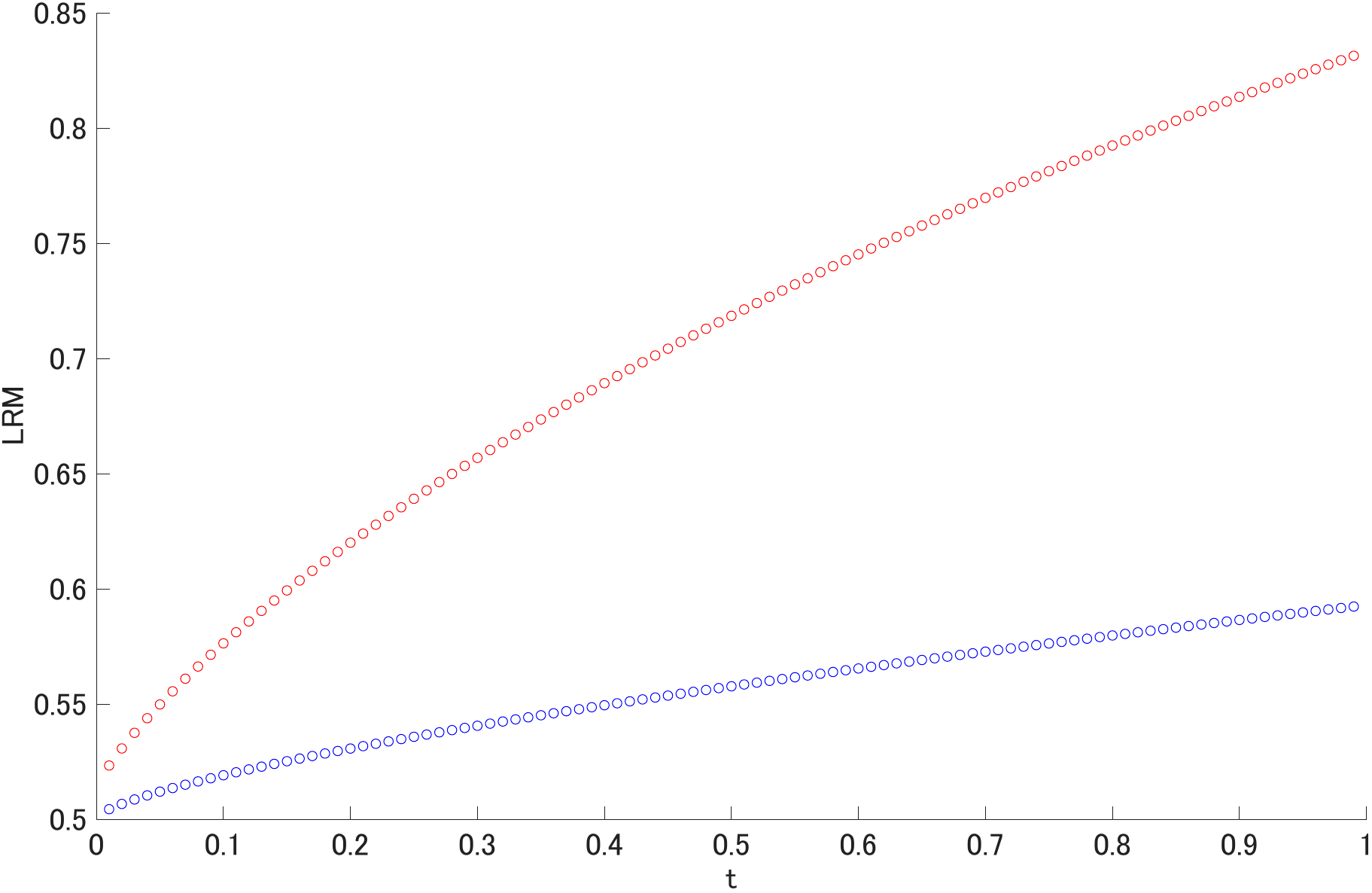}\subcaption{ \ Experiment (A) for Model (1)}\end{minipage}
 \renewcommand{\thesubfigure}{A2} \begin{minipage}{0.5\hsize}\includegraphics[width=70mm]{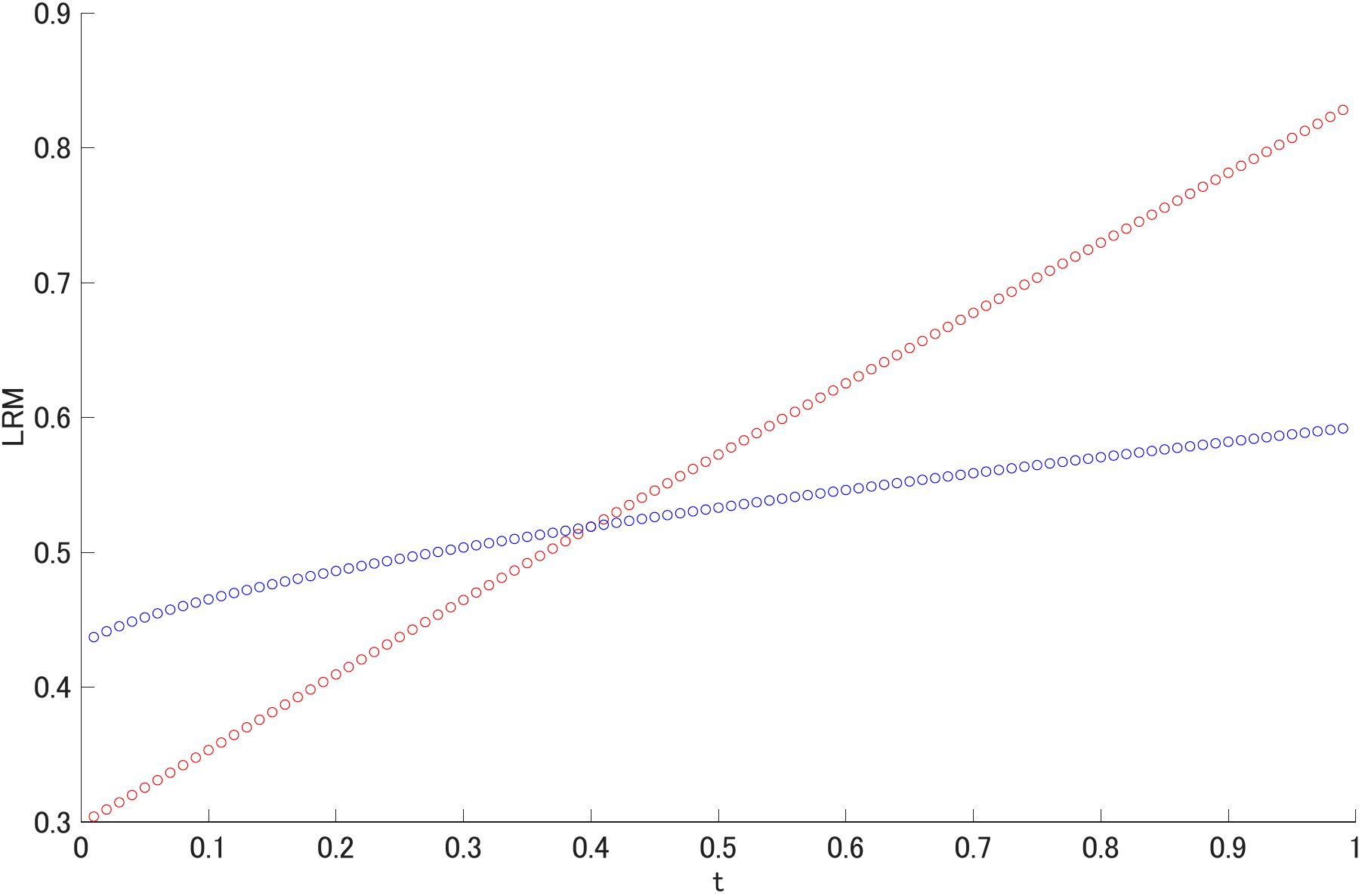}\subcaption{ \ Experiment (A) for Model (2)}\end{minipage} \\ \vspace{5mm} \\
 \renewcommand{\thesubfigure}{B1} \begin{minipage}{0.5\hsize}\includegraphics[width=70mm]{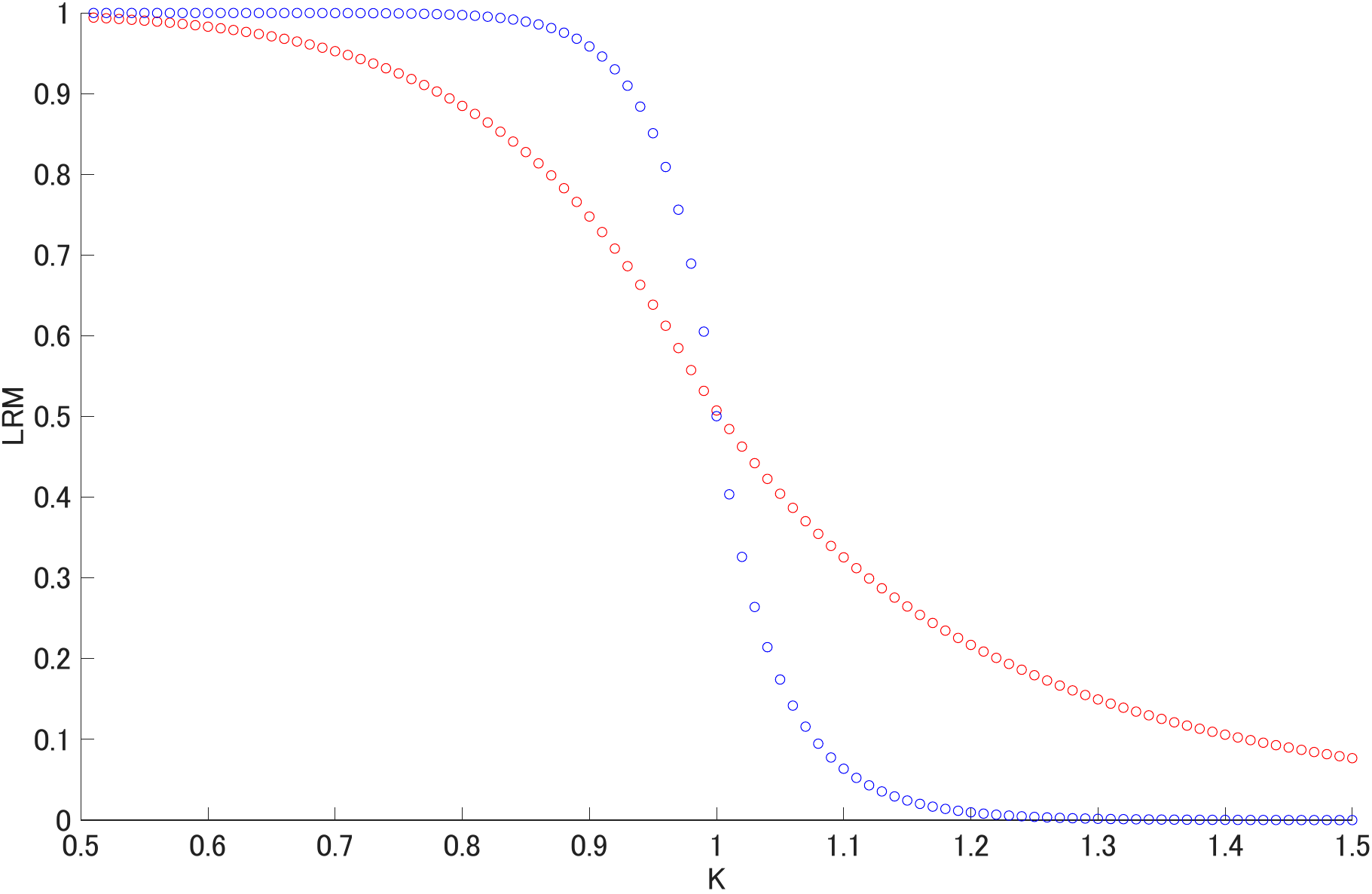}\subcaption{ \ Experiment (B) for Model (1)}\end{minipage} 
 \renewcommand{\thesubfigure}{B2} \begin{minipage}{0.5\hsize}\includegraphics[width=70mm]{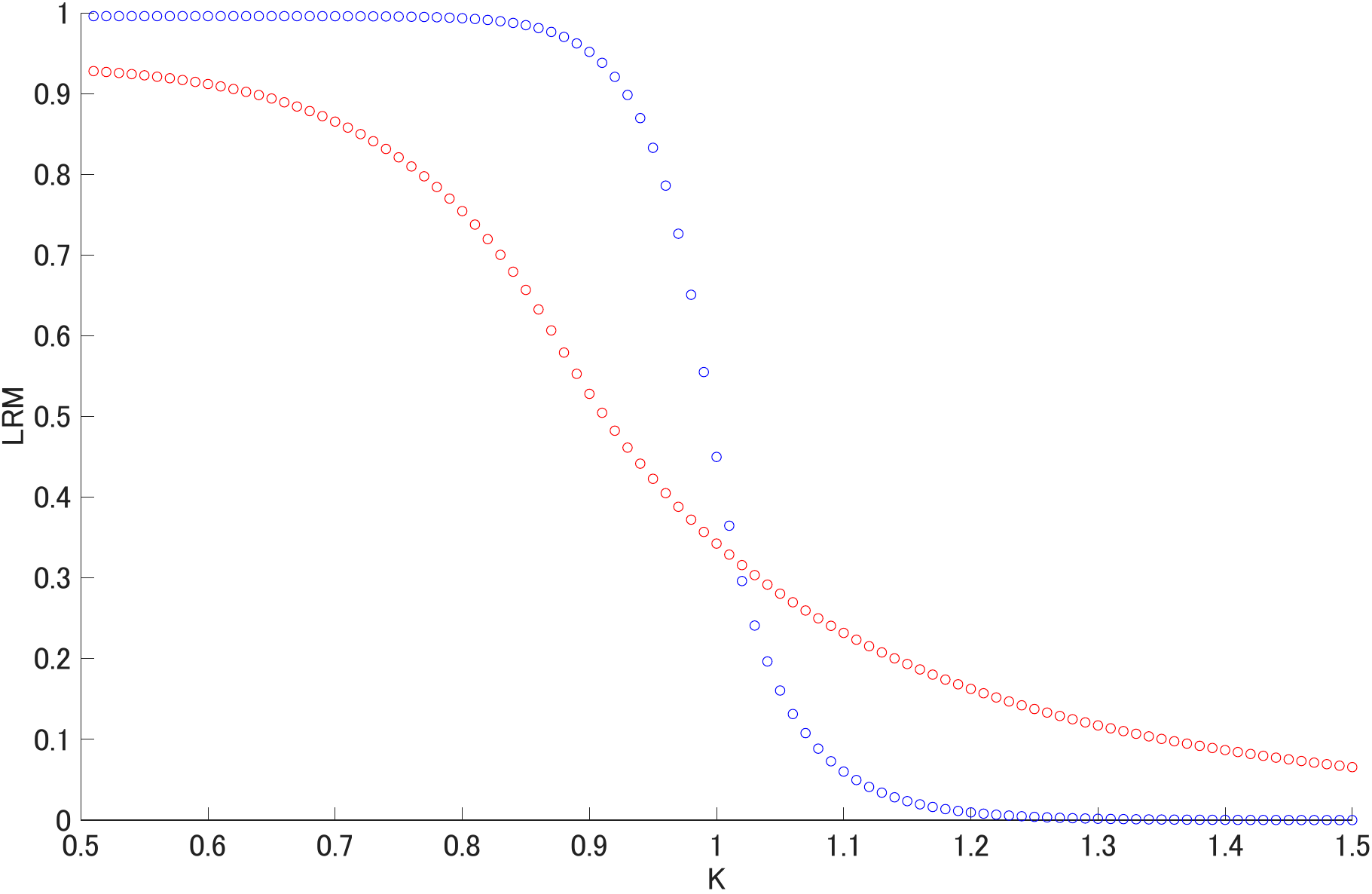}\subcaption{ \ Experiment (B) for Model (2)}\end{minipage} 
\caption{}\label{fig1}
\end{figure}

%
%
\section{Conclusion}\setcounter{equation}{0}
The goal of this paper is to study LRM strategies for exponential additive models. In fact, we derived a mathematical expression and its numerically tractable form.
In particular, to derive the integrability of the asset price process and the Radon-Nikodym density of the MMM,
we employed a different approach from that used in the L\'evy process case.
Furthermore, in Section 5, we introduced the VGSSD process and conducted numerical experiments.
As future work, numerical experiments concerning other exponential additive models remain.


\end{document}